\PassOptionsToPackage{nameinlink}{cleveref}

\documentclass[a4paper,UKenglish,cleveref,autoref,thm-restate,nolineno]{socg-lipics-v2021}

\hideLIPIcs  

\usepackage{todonotes}
\usepackage{pifont}

\usepackage{amsmath}
\usepackage{bbm}

\usepackage{algorithm}
\usepackage[noend]{algpseudocode}

\usepackage[dvipsnames]{xcolor}

\newtheorem{problem}{Problem}
\newtheorem{fact}{Fact}

\crefname{remark}{remark}{remarks}
\Crefname{remark}{Remark}{Remarks}

\crefname{obs}{observation}{observations}
\Crefname{obs}{Observation}{Observations}
\crefname{fact}{fact}{facts}
\Crefname{fact}{Fact}{Facts}

\usepackage{cancel}
\usepackage[normalem]{ulem}

\title{A Tour of Locality Sensitive Filtering on the Sphere}

\author{Luca Becchetti}{\emph{Sapienza} University of Rome, Rome, Italy}{becchetti@diag.uniroma1.it}{https://orcid.org/0000-0002-4941-0532}{}

\author{Andrea Clementi}{\emph{Tor Vergata} University of Rome, Rome, Italy}{clementi@mat.uniroma2.it}{https://orcid.org/0000-0002-9521-2457}{}

\author{Luciano Gual\`a}{\emph{Tor Vergata} University of Rome, Rome, Italy}{guala@mat.uniroma2.it}{https://orcid.org/0000-0001-6976-5579}{}

\author{Emanuele Natale}{CNRS, Universit\'e C\^{o}te d'Azur, France}{emanuele.natale@univ-cotedazur.fr}{https://orcid.org/0000-0002-8755-3892}{}

\author{Luca {Pep\`e Sciarria}}{\emph{Tor Vergata} University of Rome, Rome, Italy}{luca.pepesciarria@gmail.com}{https://orcid.org/0000-0003-4432-6099}{}

\author{Alessandro Straziota}{\emph{Tor Vergata} University of Rome, Rome, Italy \and \url{https://alessandrostr95.github.io/alessandro-straziota/}}{alessandro.straziota@uniroma2.it}{https://orcid.org/0009-0008-4543-786X}{}

\authorrunning{L. Becchetti, A. Clementi, L. Gual\`a, E. Natale, L. Pep\`e~Sciarria, and A. Straziota} 

\Copyright{Luca Becchetti, Andrea Clementi, Luciano Gual\`a, Emanuele Natale, Luca Pep\`e~Sciarria, and Alessandro Straziota} 



\category{} 

\relatedversion{} 

\EventEditors{John Q. Open and Joan R. Access}
\EventNoEds{2}
\EventLongTitle{42nd Conference on Very Important Topics (CVIT 2016)}
\EventShortTitle{CVIT 2016}
\EventAcronym{CVIT}
\EventYear{2016}
\EventDate{December 24--27, 2016}
\EventLocation{Little Whinging, United Kingdom}
\EventLogo{}
\SeriesVolume{42}
\ArticleNo{23}

\begin{document}
\newcommand{\Prob}[2]{\mathbb{P}_{#1} \left( #2 \right)}
\newcommand{\Expec}[2]{\mathbb{E}_{#1} \left[ #2 \right]}
\newcommand{\Var}[2]{\mathrm{Var}_{#1} \left( #2 \right)}
\newcommand{\Cov}[2]{\mathrm{Cov}_{#1} \left( #2 \right)}
\newcommand{\Corr}[1]{\mathrm{Corr}\left( #1 \right)}
\newcommand{\skproof}{\noindent\textit{Sketch of Proof. }}
\newcommand{\ideaproof}{\noindent\textit{Idea of the proof. }}

\newcommand{\polylog}[1]{\mbox{polylog}\left(#1\right)}
\newcommand{\poly}{ {\mathrm{poly}}}
\newcommand{\real}{\mathbb{R}}
\newcommand{\bigO}{\mathbf{O}}

\newcommand{\eas}{Spherical Filtering}
\newcommand{\sfd}{\textsc{sf data structure}}
\newcommand{\sfq}{\textsc{sf query algorithm}}

\newcommand{\minhash}{MinHash}
\newcommand{\lsf}{LSF}
\newcommand{\lsh}{LSH}
\newcommand{\anns}{ANNS}
\newcommand{\buck}{B}   
\newcommand{\topk}{\text{top-k}}      

\newcommand{\andy}{\textbf{ANDY:} }
\newcommand{\ema}[1]{\textcolor{red}{E: #1}}
\newcommand{\lb}[1]{\textcolor{blue}{LB: #1}}

\newcommand{\intersez}{\iota}
\newcommand{\Intersez}{\mathcal I}

\newcommand{\dset}{\mathcal{X}} 
\newcommand{\dist}{\textit{dist}}   
\newcommand{\lshfam}{\mathcal{F}}  
\newcommand{\pmat}{A} 
\newcommand{\sgn}{\mathbf{sgn}} 
\newcommand{\rp}{\Theta} 
\newcommand{\hash}{h} 
\newcommand{\ghash}{\mathbf{h}} 
\newcommand{\sig}{\text{sig}} 
\newcommand{\thresh}{X} 
\newcommand{\ball}{\mathbf{B}} 
\newcommand{\scap}{C} 
\newcommand{\atp}{p} 
\newcommand{\jatp}{q} 
\newcommand{\event}{\mathcal{E}} 

\newcommand{\norm}{\mathcal{N}} 

\newcommand{\nlin}{f} 
\newcommand{\eve}{\mathcal{E}} 

\def\bzero{{\bf 0}}
\def\bone{\mathbf{1}}
\def\be{{\bf e}}
\def\bx{{\mathbf x}}
\def\by{{\mathbf y}}
\def\bs{{\mathbf s}}
\def\bz{{\mathbf z}}
\def\bw{{\mathbf w}}
\def\bu{{\mathbf u}}
\def\bv{{\mathbf v}}
\def\ba{{\mathbf a}}
\def\bq{{\mathbf q}}


\maketitle

\begin{abstract}
  The Approximate Near Neighbor (ANN) problem is a cornerstone of high-dimensional data analysis. While Locality Sensitive Hashing (LSH) has been the classical paradigm, recent work has investigated Locality Sensitive Filtering (LSF), which can afford greater expressivity by allowing asymmetric regions to independently control queries and data updates.
  In its full generality, however, this framework can obscure the essential algorithmic ideas beneath substantial technical complexity.

  In this work, we bridge classical \lsh\ and modern \lsf\ by providing a self-contained, streamlined treatment of angular distance on the unit sphere under symmetric Gaussian filters.
  Although specialized,
  this canonical setting -- monotonically equivalent to the widely used cosine similarity -- captures all key aspects of the problem, delivering a ``guided tour'' of the core mechanisms of locality-sensitive filtering and allowing us to expose all fundamental algorithmic and probabilistic ingredients, while retaining transparent notation and pedagogical clarity.

  
  
  More precisely, we develop a treatment of Spherical-\lsf\ that makes explicit the connection between the probabilistic behavior of Gaussian filters and the performance of the resulting data structure.
  Using elementary calculus, we derive sharp probability bounds and obtain a transparent proof of asymptotic optimality among schemes whose filter distribution is fixed independently of the input dataset.
  We further exploit the specific structure of Gaussian filters to obtain an efficient implementation while avoiding much of the technical machinery required by the general \lsf\ framework.

  Overall, our results offer a self-contained and accessible entry point to locality-sensitive filtering, while providing a unified view of its main probabilistic and algorithmic principles.
\end{abstract}

\section{Introduction}\label{se:intro}

In the \emph{Nearest Neighbour} problem, we are given a dataset $\dset$ of points from a metric space $(\mathcal{M}, \dist)$. The goal is to preprocess $\dset$ so as to efficiently answer queries of the form: given a point $\bq \in \mathcal{M}$, retrieve a point $\bx \in \dset$ that is among the closest to $\bq$.
Despite its seeming simplicity, Nearest Neighbour entails a fundamental primitive in a wide range of modern applications -- including machine learning, data mining, information retrieval, and computer vision -- in which data are naturally represented as vectors in a feature space. In this setting (often referred to as the \emph{vector space model}), similarity between objects is typically captured by a distance function, and identifying the most similar object reduces precisely to solving a nearest neighbor query. Canonical examples include similarity search in image databases, nearest-neighbor classification, and content recommendation.

This problem poses significant computational challenges when data volume grows to a scale. In this case, naive solutions with linear query-time and/or quadratic space are no longer feasible.
Moreover, known sublinear solutions exhibit an exponential dependence on the dimensionality of the data in the query time and/or space \cite{arya-jacm93,arya-jacm98,chanSCG97,clarkson,kleinberg-nn}, the well-known \emph{curse of dimensionality}.

These limitations motivate the study of approximate variants of the problem.
In a seminal work, Indyk and Motwani \cite{har2012approximate} introduced \emph{Locality-Sensitive Hashing} (\lsh), 
a randomized framework that allows to \emph{approximate} the Nearest Neighbour\footnote{Where we want to retrieve any data point whose distance is at most $c$ times the distance of the nearest.} with polynomial preprocessing time, sublinear query time, and subquadratic space.
This approach leverages a reduction to the \emph{$(r,cr)$-Approximate Near Neighbor} (or $(r,cr)$-ANN) problem in which, given a query $\bq$, we must return a point $\bx' \in \dset$ at distance at most $cr$ with constant probability, whenever a point $\bx$ at distance at most $r$ from $\bq$ exists in $\dset$.
The idea behind \lsh\ is to use hash functions that are \emph{sensitive} with respect to the metric, i.e., such that points that are close are more likely to share the same hash value.
More formally, we say that a family $\mathcal{H}$ of hash functions is $(r,cr, p_1, p_2)$-\emph{sensitive} for a metric space $(\mathcal{M}, \dist)$ if, for a random $h \in \mathcal{H}$, and for every $\bx,\by \in \mathcal{M}$, we have $\Prob{}{h(\bx) = h(\by)} \geq p_1$ if $\bx,\by$ are at distance at most $r$ (we say that $\bx$ and $\by$ are \emph{close} points), and $\Prob{}{h(\bx) = h(\by)} \leq p_2$ if $\bx,\by$ are at distance at least $cr$ (we say that $\bx$ and $\by$ are \emph{far} points).

Indyk and Motwani \cite{indyk98lsh} proved that a $(r,cr, p_1,p_2)$-sensitive hash family $\mathcal{H}$ for $(\mathcal{M}, \dist)$ affords a data structure that solves $(r,cr)$-ANN with constant error probability  for a query point $\bq$, in time $O(n^{\rho})$ and using $O(n^{1+\rho})$ space, where $\rho = \frac{\log\frac{1}{p_1}}{\log\frac{1}{p_2}}$.
This solution not only bypasses the curse of dimensionality, but it also achieves sublinear query time and subquadratic space whenever $p_1 > p_2$.
A vast literature has developed on \lsh, leading to the study of \lsh\ functions for various metrics, including Hamming \cite{har2012approximate}, Jaccard \cite{broder1997resemblance,broder2000identifying}, angular \cite{andoni2015practical,charikar2002similarity,sphericalLSH}, and Euclidean \cite{andoni2008near,andoni2014beyond,andoni2015optimal} distances.

Since $\rho$ characterizes the performance of the data structure, the problem essentially boils down to designing ever more ``sensitive'' hash functions.
In particular, the best sensitive families known so far yield $\rho \leq \frac{1}{c} + o_d(1)$ for Hamming and Jaccard distances, while $\rho \leq \frac{1}{c^2} + o_d(1)$ for Euclidean \cite{andoni2006near,andoni2008near} and angular distances \cite{andoni2015practical,sphericalLSH}\footnote{Here $o_d(1)$ denotes a function that, asymptotically in $d$, is smaller than any constant.}.
This implies that, as the required approximation loosens (i.e., parameter $c$ increases), increasingly efficient solutions are possible.
On the other hand, the upper bounds on $\rho$ above are essentially tight \cite{motwani2006lower,o2014optimal}.

The \lsh\ framework is not the only approach to the $(r,cr)$-ANN problem. A related, yet different approach is the more recent \textit{Locality Sensitive Filtering} (\lsf), introduced in \cite{becker2016new}. One of the main advantages of \lsf\ is that it potentially offers a more flexible and agile way to tune the behavior of the data structure and the underlying space-time trade-offs \cite{christiani2017framework}. Informally, a \textit{filter} $h$ is a boolean function, and we say that a data point $\bx$ passes the filter if $h(\bx) = 1$.
In the spirit of \lsh, in \lsf\ one designs a family of filters $\mathcal{F}$ that is \emph{sensitive} with respect to a desired distance, so that points that are close enough are more likely to pass the same filters than points that are farther apart. More formally, a family of filters $\mathcal{F}$ is $(r,cr,q_1,q_2)$-\textit{sensitive} for a metric space $(\mathcal{M}, \dist)$ if, for a randomly chosen filter $h \in \mathcal{F}$, and for every $\bx,\by \in \mathcal{M}$, we have  $\Prob{}{h(\bx) = 1 \mid h(\by) = 1} \geq q_1$ if $\bx$ and $\by$ are ``close'', and $\Prob{}{h(\bx) = 1 \mid h(\by) = 1} \leq q_2$ if $\bx$ and $\by$ are ``far apart''.
\lsh\ and \lsf\ seemingly offer different perspectives on local sensitivity: While in \lsh\ we are interested in the joint probability that two close points share the same hash value, in \lsf\ our focus is on the probability that two points pass the same filter.
In other words, while a hash function induces a partition of the input space in \lsh, a filter does not provide any additional information regarding points that did not pass a filter.
Similarly to \lsh, one can associate a parameter $\rho=\frac{\log(1/q_1)}{\log(1/q_2)}$ to an \lsf\ family with probabilities $q_1,q_2$, which measures the discriminating power of the filters and determines the performance of the corresponding data structure. 
Connections between \lsh\ and \lsf\ run deep and are reflected in a relationship between the corresponding versions of parameter $\rho$,  something we discuss in a later section.

In this work, we consider the ANN problem when $\mathcal{M}$ is the unit $d$-dimensional sphere $S^{d-1}$ equipped with the \emph{angular distance}, a variant we  refer to as the \emph{Angular ANN} problem.
This variant is highly relevant in practice, since angular distance is equivalent to  Euclidean distance on the sphere and thus cosine similarity, commonly adopted in applications including image search \cite{jegou2010}, speaker representations \cite{schmidt2014} and tf-idf datasets \cite{sundaram2013}. Moreover, all results for Angular ANN seamlessly extend to the entire Euclidean space
\cite{aumuller2019puffinn,christiani2017framework,rahimi2007random,valiant2015finding}.

\subparagraph*{Scope of the paper.}
The locality-sensitive filtering framework is considerably more general than the setting considered in this paper.
In its general form, it applies to arbitrary similarity or dissimilarity spaces and allows a filter to use different regions for query and update operations.
This asymmetry gives rise to two distinct exponents, $\rho_q$ and $\rho_u$, controlling query time and space complexity, respectively.

Our goal is not to redevelop this framework in its full generality. Instead, we deliberately focus on the canonical geometric setting of the unit sphere equipped with angular distance, and on symmetric filters, for which query and update regions coincide and the two exponents collapse to a single parameter $\rho$.
This special instance of \lsf\ allows us to isolate the central probabilistic and algorithmic ideas behind \lsf\ while keeping the notation and data structure transparent.
We return to asymmetric filters and the resulting space-time tradeoffs in \Cref{sec:asymmetric_lsf}.

\subsection{Our Contribution}\label{subse:contr}
In this paper, we give a self-contained and deliberately focused presentation of locality-sensitive filtering for Angular ANN.
Rather than introducing the \lsf\ framework in its full generality, we isolate the case of symmetric Gaussian filters on the unit sphere.
This setting is sufficiently rich to expose all fundamental ingredients of \lsf, while enabling a presentation and analysis that require minimal technical overhead.
We offer three main contributions:
\begin{description}
    \item[An elementary data structure:] We first introduce a basic Spherical-\lsf\ data structure in which every sampled Gaussian filter corresponds to a single bucket.
    On the positive side, this basic data structure presents a clean connection between filter and success probabilities, expected number of candidates and space complexity.
    Its direct implementation, however, may be inefficient, requiring the evaluation of too many filters at query time.
    To address this issue, we refine the elementary data structure above using the tensoring technique of Christiani \cite{christiani2017framework}, obtaining a more refined data structure that allows enumeration of all buckets relevant to a query in (optimal) expected time $O(n^\rho)$.

    \item[A simplified use of tensoring:] In the general \lsf\ framework \cite{christiani2017framework}, tensoring is combined with another technique, called powering, to obtain a required probability amplification, while at the same time affording efficient filter enumeration. For the Gaussian filters considered here, the threshold parameter itself acts as a probability amplifier. This allows us to completely do away with powering and to provide a relatively simple tensorized data structure and analysis.



    \item[A streamlined analysis of Gaussian filters:] A central technical ingredient in the literature on Angular (and Euclidean) ANN is bounding the probability that the correlated Gaussian projections resulting from application of the same filter to two different input points simultaneously exceed a given threshold.
    While a cornerstone for both angular \lsf\ and \lsh, previous solutions to this problem have been mathematically involved, asymptotic in nature, or incomplete. 
    For instance, the analysis of \cite{christiani2017framework} does not cover all angle ranges and relies on coarse geometric approximations (e.g., bounding the probability of a multivariate normal vector belonging to a closed convex body) that are only tight up to polynomial factors. 
    Similarly, Becker et al. \cite{becker2016new} circumvent the exact probabilistic analysis by relying on bounds for the volume of the intersection of spherical caps, 
    which requires additional non-trivial steps to relate back to Gaussian projectors. 
    In contrast, \Cref{le:cond_prob} offers a fully self-contained, elementary analysis of this problem using only basic calculus. 
    Our bounds are rigorous, hold for the full range of angles, and are slightly tighter than those in previous literature. 
    This direct approach completely does away with heavy asymptotic notation and complex geometric machinery used in previous work. 
    In turn, this simplification directly translates into a much more streamlined and transparent proof of the asymptotic
    optimality of Spherical LSF (\Cref{subse:slsf_optimal}).
\end{description}

\subparagraph{Paper organization.}
In \Cref{subse:related} we discuss related works and the main problematic of the \lsf\ literature.

\Cref{se:prelim} introduces the notation and the Angular ANN problem.
In \Cref{se:eas,se:analysis} we deliberately begin with the basic version of the Spherical-\lsf\ data structure.
This construction is useful pedagogically because it separates the probabilistic behavior of the filters from the algorithmic problem of enumerating them.
\Cref{se:analysis} analyzes the filter probabilities, proves the optimality of the resulting exponent $\rho$, and bounds the expected number of candidates and the space usage of the basic construction.

The analysis of the basic construction reveals that the number of filters that must be inspected may exceed the number of candidate points.
\Cref{se:tensoring} closes this gap by applying tensoring, thereby obtaining an efficient enumeration of the buckets associated with a query and the claimed end-to-end query complexity.
Finally, \Cref{sec:asymmetric_lsf} places our symmetric construction within the general asymmetric \lsf\ framework, explains the origin of query-space tradeoffs, and briefly discusses the stronger guarantees available in the data-dependent setting.

\subsection{Further Related Work}\label{subse:related}
Before discussing the \lsf\ literature, we briefly outline the major results on \lsh\ for the ANN problem.

\subparagraph{Locality-sensitive hashing.}
The seminal work of Indyk and Motwani \cite{indyk98lsh} introduced locality-sensitive hashing, studying \lsh\ families for Hamming distance, $\ell_1$ norm and Euclidean distance ($\ell_2$ norm).
Charikar~\cite{charikar2002similarity} subsequently introduced the \emph{random-hyperplane} \lsh\ (or \emph{SimHash}) scheme for angular distance, yielding an efficient estimator for cosine similarity.

Notable works on the Euclidean distance and angular distance are the following.
Datar et al. \cite{datar2004pstable} is the first work that provides a natural hash family for Euclidean distance without embedding it into the Hamming space, which achieves $\rho \approx 1/c$. 
Motwani et al. \cite{motwani2006lower} show the first lower bounds to the \lsh\ data structure under the $\ell_1$ norm. Later, Andoni and Indyk \cite{andoni2006near} introduce a scheme for the Euclidean distance achieving $\rho \approx 1/c^2$, significantly improving upon the previous $1/c$ bound. This result has been later shown to be optimal for data-independent LSH in the Euclidean setting \cite{chandrasekaran2018lattice,o2014optimal}.

Indeed, from the lower bound point of view the work of O'Donnell et al. \cite{o2014optimal} provides a lower bound for the Hamming distance in which $\rho \ge 1/c - o_d(1)$. By applying their results to the unit sphere or Euclidean space, they derive a lower bound for $\rho$ of $1/c^2$ in that setting.

Of particular interest is the work of Terasawa and Tanaka \cite{sphericalLSH} whose \lsh\ scheme, named \emph{cross-polytope}, is analyzed in \cite{andoni2015practical} and shown to be optimal. 
The cross-polytope \lsh\ has some resemblance to the \lsf\ scheme studied in this paper. Indeed, the main core of the cross-polytope hash function resides in the computation of dot-products between the points in the dataset and random Gaussian vectors sampled from $\norm\left(\bzero, I/d\right)$. The analysis provided here may also be applied to that particular case.
The cross-polytope scheme of \cite{andoni2015practical} has been used into practical implementations \cite{aumuller2019puffinn,pham2022falconn++}.
Finally, stronger guarantees can be obtained when the hashing scheme is allowed to depend on the input dataset. In fact, to overcome the limitations imposed by the lower bound $\rho \geq 1/c^2$, \emph{data-dependent} approaches are introduced and analyzed in \cite{andoni2014beyond,andoni2015optimal}, achieving a superior exponent of $\rho = 1/(2c^2 - 1)$ for Euclidean ANN.

\subparagraph{Locality-sensitive filtering.}
Locality-sensitive filtering was introduced by Becker et al. \cite{becker2016new} for approximate near-neighbor search on unit sphere, with applications to the Shortest Vector Problem.
In their model, filter evaluation is described through an oracle that, given a point, lists the filters containing it.
For the spherical construction developed in the same work, however, this oracle is explicitly instantiated: independently sampled filters are replaced by a structured family based on random product codes, and the relevant filters are enumerated through a corresponding list-decoding procedure.

Christiani~\cite{christiani2017framework} subsequently developed a more general data-indepedent \lsf\ framework.
The framework allows query and update operations to use different filter regions, giving rise to separate query and update exponents, and hence to space-time tradeoffs.
Moreover, it removes the oracle abstraction of \cite{becker2016new}, incorporating efficient filter enumeration into general construction through a combination of the so called \emph{powering and tensoring techniques}.
Shortly, powering amplifies the locality-sensitive properties of the original filter family, whereas tensoring generates a large, structured family of composite filters whose members containing a given point can be enumerated efficiently.

As for \lsh, a separate line of work considers data-dependent approaches, in which the indexing strategy may adapt to the input dataset.
Building on spherical filtering and data-dependent hashing, Andoni et al.~\cite{andoni2017optimal} established hashing-based
space-time tight tradeoffs for approximate near-neighbor search.
These results concern a more powerful model than the data-independent setting studied in the main body of this paper and it requires entirely different algorithmic and geometric machinery.
We discuss asymmetric filters and the distinction between data-independent and data-dependent schemes in \Cref{sec:asymmetric_lsf}.

Technical comparisons with Becker et al. \cite{becker2016new} and Christiani
\cite{christiani2017framework} are provided in \Cref{se:analysis}.

Beyond its theoretical development, locality-sensitive filtering has also inspired practical ANN methods, including FALCONN++~\cite{pham2022falconn++}.
Related locality-sensitive techniques have also found applications in the context of algorithmic fairness \cite{aumuller2022fairtods}, and differential privacy \cite{aumuller2025privatenn}.
More broadly, random-projection techniques, closely related to the Gaussian filters studied here, have been applied to clustering problems \cite{xu2024sdbscan,zheng2026robust}.



Despite these developments, locality-sensitive filtering receives only limited treatment in standard textbooks and surveys \cite{Bruch2024FoundationsOV,har-peled2011book,Leskobook20,prezza2024algorithmsmassivedata,toth2017handbook}.
For example, the survey of Andoni, Indyk, and Razenshteyn \cite{andoni2018survey} mentions the framework only briefely.

\section{Preliminaries}\label{se:prelim}

\subparagraph*{Notation.}
For any non-negative integer $m$, let $[m] = \{1,\ldots , m\}$. In general, we use bold symbols to represent vectors. In particular, we use $\be_i$ to denote the $i$-th canonical vector. We further let $S^{d-1} = \{\bx\in\real^d: \|\bx\| = 1\}$, i.e., the $d$-dimensional unit sphere. We use $I_d$ to denote the identity matrix in $d$ dimensions, omitting $d$ when clear from context.
In the remainder, $\norm(\bzero, I_d)$ denotes the multivariate, standard normal distribution, while we respectively use $\phi$ and $\Phi$ to denote the density and CDF of its marginals \cite{bishop2006pattern}, namely
\begin{equation*}
    \phi(x) = \frac{1}{\sqrt{2\pi}}e^{-\frac{1}{2}x^2}, \ \ \ \ \Phi(x) = \int_{-\infty}^x\phi(x) dx.
\end{equation*}
Finally, considered the set of functions that include a parameter $z$, we denote by $o_z(1)$ the subset of functions that tend to $0$ as $z\rightarrow\infty$.

\subparagraph*{Useful facts.}
The following well-known results will be extensively used in the remainder.
\begin{lemma}[Proposition 2.1.2 in \cite{vershynin2018high}]\label{le:normal_tail_1}
Assume $g \sim \mathcal{N}(0,1)$. Then, for every $t > 0$
    \[
    \frac{t}{t^2 + 1} \cdot \frac{1}{\sqrt{2\pi}} e^{-t^2/2} \leq \Prob{}{g \geq t} \leq \frac{1}{t} \cdot \frac{1}{\sqrt{2\pi}} e^{-t^2/2} .
    \]
\end{lemma}

\begin{fact}\label{fa:expec}
Assume that $x\sim\norm(0, 1)$ and let $t > 0$. Then:
    \[
        \int_t^{\infty}(x - t)\phi(x)dx = \phi(t) - t(1 - \Phi(t)).
    \]
\end{fact}
The proof of the following  folklore fact is given in \Cref{se:apx_prel} for the sake of completeness.
\begin{fact}\label{fa:standard_ratio}
For any $0\le\eta_1\le\eta_2$ we have:
\begin{equation}\label{eq:standard_ratio}
    \frac{1 - \Phi(\eta_2)}{1 - \Phi(\eta_1)}\le e^{-\frac{\eta_2^2 - \eta_1^2}{2}}.
\end{equation}
\end{fact}

\subparagraph*{Problem.}
In this paper, we are interested in the $(r, cr)$-Approximate Near Neighbor problem when the underlying metric space $(\mathcal{M}, \dist)$ is $S^{d-1}$ equipped with the angular distance. 
\begin{problem}[Angular ANN]\label{prob:r_cr}
    Consider the metric space $(S^{d-1}, \angle(\cdot, \cdot))$ where $\angle$ is the angular distance, a dataset $\dset \subseteq S^{d-1}$, any given $\gamma > 0$, $c > 1$ and $\delta \in (0,1)$. For any point $\bq \in S^{d-1}$, the Angular $(\gamma, c\gamma)$-\textit{ANN} problem asks to return a point $\bx' \in \dset$ such that $\angle(\bq, \bx') \leq c\gamma$ if there exists at least one point $\bx \in \dset$ such that $\angle(\bq, \bx) \le \gamma$, with probability at least $1-\delta$.
\end{problem} 
\begin{remark}\label{rem:distance}
    We state \Cref{prob:r_cr} in terms of angular distance since it makes presentation of some technical aspects smoother. On the other hand, most literature on the Angular ANN problem equivalently considers the Euclidean distance between points on the unit sphere (e.g., \cite{andoni2015practical,andoni2015optimal}). It is worth recalling that two points on $S^{d-1}$ sharing an angles $\alpha$ have Euclidean distance $r = 2\sin\frac{\alpha}{2}$. We revisit the impact of this connection on the form of some results in the discussion at the end of \Cref{se:analysis}.
\end{remark}

\section{Locality Sensitive Filters for Angular Distance}\label{se:eas}

In this section, we introduce a basic version of the
\emph{Spherical-\lsf\ data structure} for the angular
$(\gamma,c\gamma)$-ANN problem.
Our presentation proceeds in two steps.
First, in \Cref{sse:LSF}, we introduce locality-sensitive filtering and its main performance parameter.
Then, in \Cref{ssec:slsf}, we specialize the framework to the unit sphere equipped with angular distance and describe the resulting indexing and query procedures.
The analysis of this basic construction is deferred to \Cref{se:analysis}.
Although this construction already captures the probabilistic and algorithmic mechanisms underlying Spherical-\lsf, its direct implementation may require evaluating a large number of filters.
We address this algorithmic bottleneck in \Cref{se:tensoring}.

\subsection{Locality-Sensitive Filters}\label{sse:LSF}
We next provide the rigorous definition of a locality-sensitive filter family for the angular distance on the sphere, since this is the case we are working with.
\begin{definition}[Spherical \lsf]\label{def:lsf}
Consider the sphere $S^{d-1}$, any angle $0\le\gamma\le\pi$ and $c > 1$. For two values $0 < q_2 < q_1 < 1$, a family $\lshfam = \{ h:S^{d-1} \to \{0,1\} \}$ of filters is $(\gamma,c\gamma,q_1,q_2)$-locality sensitive if, for $\hash$ sampled uniformly at random from $\lshfam$ and for any two points $\bx, \bq \in S^{d-1}$ we have:
\begin{itemize}
    \item $\Prob{\hash\sim\lshfam}{\hash(\bx) = 1\mid \hash(\bq) = 1}\ge q_1$ if $\angle(\bq, \bx)\le \gamma$;
    \item $\Prob{\hash\sim\lshfam}{\hash(\bx) = 1\mid \hash(\bq) = 1}\le q_2$ if $\angle(\bq, \bx)\ge c\gamma$.
\end{itemize}
We will also refer to $\lshfam$ as spherical \lsf\ family, or simply as \lsf\ family.
\end{definition}

Along the lines of \lsh, a family of filters is as powerful as $q_1$ is large and $q_2$ is small. A convenient and compact way to define its discriminating power is through the parameter $\rho$. 

\begin{definition}\label{def:rho}
Given the probabilities $q_1$ and $q_2$ from \Cref{def:lsf}, we define the parameter $\rho$ of an \lsf\ family $\lshfam$ as 
    \begin{equation}
    \label{eq:def_rho}
        \rho = \frac{\log\frac{1}{q_1}}{\log\frac{1}{q_2}}.     
    \end{equation}
\end{definition}
Parameter $\rho$ provides a compact measure of the discriminating power of an \lsf\ family.
As we show in \Cref{se:analysis}, it determines the performance of
the corresponding \lsf-based data structure.
The relationship between this parameter and the analogous exponent
for \lsh, together with the resulting lower bounds, is discussed in
\Cref{subse:slsf_optimal}.

\subsection{The Basic Spherical-LSF data structure}\label{ssec:slsf}

We begin the description of the data structure for the angular $(\gamma, c\gamma)$-ANN problem by describing the \lsf\ family we use for the angular distance, also known as Gaussian filters, which is based on the definition of suitable spherical caps over the unit sphere. 

\begin{definition}[Gaussian filters]
Let $\tau \geq 0$ be a threshold parameter.
A filter $h$ from the \lsf\ family $\lshfam_\tau$ is defined as 
\begin{equation*}
h(\bx) = \begin{cases}
   1 &\text{if } \theta^\top\bx \geq \tau\\
   0 &\text{otherwise}
\end{cases},
\end{equation*}
where $\theta \sim \norm(\bzero, I/d)$, i.e., the $d$-dimensional vector where each entry is an independent Gaussian random variable with mean $0$ and variance $\frac{1}{d}$.
\end{definition}

\begin{remark}\label{rem:symmetry}
In the specific scenario considered in this paper (which is standard in the literature) the following symmetry properties hold: (i) $\Prob{}{\hash(\bx) = 1}$ \emph{does not depend} on $\bx$ (though it will depend on the parameter $\tau$) and
(ii) for a fixed threshold $\tau$, the joint probability $\Prob{}{\hash(\bx) = \hash(\bq)=1}$ depends only on $\angle(\bx,\bq)$, and it is non-increasing on this value.
Indeed, the two normalized projections, $\sqrt{d}\theta^\top\bq$ and $\sqrt{d}\theta^\top\bx$, are jointly Gaussian with correlation $\cos(\angle(\bx,\bq))$, and their joint upper-tail probability is increasing in the correlation.
Consequently, the boundary angles $\gamma$ and $c\gamma$ determine valid near- and far-pair probabilities for the entire corresponding distance ranges.
Finally, notice that (i) and (ii) obviously imply $\Prob{}{\hash(\bx) = 1\mid\hash(\bq) = 1} = \Prob{}{\hash(\bq) = 1\mid\hash(\bx) = 1}$.
\end{remark} 

Consider an input dataset $\dset$ of $n$ points in $S^{d-1}$. To index $\dset$ using Spherical-\lsf, we first sample $m$ independent vectors $\theta_1, \dots, \theta_m \sim \norm(\bzero, I/d)$, which yield  $m$ independent filters $h_1, \dots, h_m$ from $\lshfam_\tau$.

Next, for every $\bx \in \dset$, we compute its \emph{signature} $\sig(\bx) = \{i \in [m] \mid h_i(\bx) = 1\}$, i.e., the set of filter indices that $\bx$ passes 
(see \Cref{alg:lsf}).

The resulting data-structure is a collection of $m$ buckets $B_1, \dots, B_m$, where each $B_i$ contains all the points that pass the $i$-th filter
(see \Cref{alg:lsf_indexing}).

Finally, to process a query $\bq\in S^{d-1}$, we first compute its signature $\sig(\bq)$, then we look for a close point by scanning all buckets $B_i$, such that $i \in \sig(\bq)$ 
(see \Cref{alg:lsf_query}).


\noindent
\begin{minipage}[t]{0.46\textwidth}
    \begin{algorithm}[H]
      \caption{Spherical-\lsf-Filter}\label{alg:lsf}
      \textbf{Input:} a vector $\bx\in S^{d-1}$, a threshold $\tau\ge 0$, $m$ vectors $\theta_1, \dots, \theta_m \sim \norm(\bzero, I/d)$\\
      \textbf{Output} a set of indices $\sig(\bx)\subseteq [m]$
      \begin{algorithmic}[1]
        \State $\sig(\bx) = \{ i \in [m] \mid \theta_i^\top\bx \geq \tau \}$\;
        \State \Return $\sig(\bx)$\;
      \end{algorithmic}
    \end{algorithm}%
    \begin{algorithm}[H] 
    \caption{Spherical-\lsf-Index}\label{alg:lsf_indexing}
        \textbf{Input:} signatures $\{\sig(\bx): \bx\in\dset\}$\\
        \textbf{Output} a collection of buckets $\buck_1,\ldots, \buck_m$
        \begin{algorithmic}[1]
            \For{$i = 1$ to $m$}
                \State $\buck_i = \{\bx \in \dset \mid i \in \sig(\bx) \}$\;
            \EndFor
            \State \Return $\buck_1,\ldots, \buck_m$
        \end{algorithmic}
    \end{algorithm}
\end{minipage}
\hfill 
\begin{minipage}[t]{0.46\textwidth}
    \begin{algorithm}[H]
    \caption{Spherical-\lsf-Query}\label{alg:lsf_query}
        \textbf{Input:} a query $\bq$, buckets $\buck_1, \dots, \buck_m$\\
        \textbf{Output:} a near neighbor of the query $\bq$, or \texttt{not~found}
        \begin{algorithmic}[1]
            \State compute $\sig(\bq)$\;
            \For{$i \in \sig(\bq)$}
                \For{$\bx\in\buck_i$}
                    \If{$\angle(\bx, \bq)\le c\gamma$}
                        \State \Return $\bx$\;
                    \EndIf
                \EndFor
            \EndFor
            \State \Return \texttt{not~found}
        \end{algorithmic}
    \end{algorithm}
\end{minipage}\bigskip

The geometric intuition behind this scheme is that every filter $h_i$ represents a \emph{spherical cap} of $S^{d-1}$, with $B_i$ containing all input points that belong to the same spherical cap (see \cref{fig:placeholder}).
The idea of using spherical caps (to partition the sphere) in \lsh\ was first introduced in \cite{andoni2015optimal}, where it was called Spherical LSH.
The filtering scheme proposed here is essentially equivalent to that of \cite{becker2016new}, up to their choice of using vectors distributed uniformly on $S^{d-1}$ as the projectors. It is also a special case of the filtering scheme from \cite{christiani2017framework}. We discuss these aspects and their implications in more detail after \Cref{le:cond_prob}.

\begin{figure}[h]
    \centering
    \caption{Figure (a): an example of spherical filters. Blue cap corresponds to $\{\bx\in S^{d-1} \mid \theta_i^\top\bx \geq \tau\}$, while the red one corresponds to $\{\bx \in S^{d-1} \mid \theta_j^\top\bx \geq \tau\}$. Figure (b): pipeline of Spherical-\lsf\ data structure.}
    \includegraphics[width=1\linewidth]{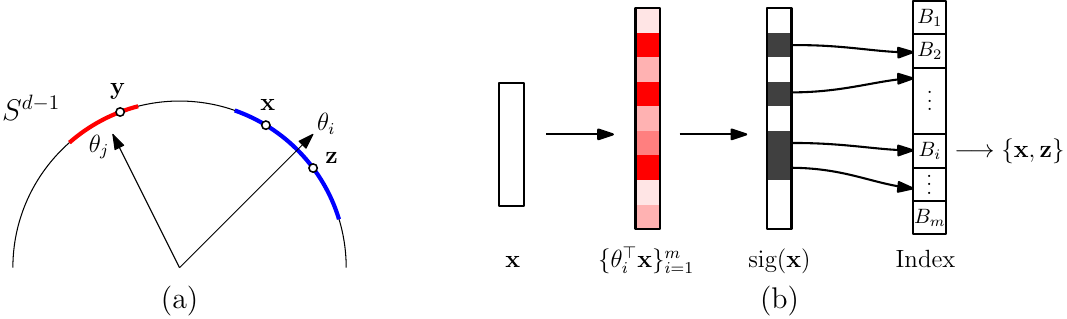}
    \label{fig:placeholder}
\end{figure}



\section{Analysis of Spherical \lsf}\label{se:analysis}
This section analyzes the performance of the data structure introduced in Section \ref{se:eas}. In most literature on the topic, this analysis addresses three aspects:
(i) a characterization of the locality sensitive properties of the proposed scheme, which in our case means possibly tight bounds on the (conditional) probabilities from \Cref{def:lsf};
(ii) a characterization of the resulting parameter $\rho$ (\Cref{def:rho}), measuring the discriminating power of the proposed scheme;
(iii) a quantitative analysis of the space and time requirements of the proposed data structure when solving the $(\gamma, c\gamma)$-ANN problem.
These aspects are respectively covered in the following Sections \ref{subse:cond_prob}, \ref{subse:slsf_optimal} and \ref{se:eas_anns}.
Each of these sections revisits and/or refines ideas that have been introduced or adapted in previous  \lsh\ or \lsf\ literature, presenting them in a unified manner. In each of these sections, we briefly discuss key previous results that are relevant, pointing to differences in our approach or analysis when this is the case.

\subsection{Analysis of collision probability}\label{subse:cond_prob}
The main technical contribution of this subsection is tight bounds on the conditional probability \(\Prob{}{h(\bx) = 1 \mid h(\by) = 1}\), which will then be used to prove optimality of \eas\ in \Cref{subse:slsf_optimal}.

Bounding $\Prob{}{h(\bx) = 1 \mid h(\by) = 1}$ is a crucial technical ingredient in all key contributions in \lsh\ and \lsf\ for Euclidean distance on the sphere, in particular \cite{andoni2015practical,andoni2014beyond,andoni2015optimal,  becker2016new,christiani2017framework}. In the remainder, we give an elementary proof that, while proceeding similarly to \cite{andoni2015optimal} (which focuses on a specific value of the threshold), is self-contained and more streamlined, avoiding asymptotic notation and providing results that are slightly tighter and that can be compactly expressed in terms of $1 - \Phi(x)$. We briefly discuss analogous or related proofs given in the aforementioned references at the end of this subsection.

\begin{lemma}
\label{le:cond_prob}
Assume $\bx,\by\in S^{d-1}$ share an angle \(\alpha\) and consider the filter $\hash(\bx) = \mathbbm{1}\{\theta^\top \bx \geq \tau\}$, where $\theta = (\theta_1,\ldots , \theta_d)^\top \sim \norm\left(\bzero, I/d\right)$ and $\tau$ satisfies $\tau\sqrt{d}\tan\frac{\alpha}{2}\ge 1$. Let $t = \sqrt{d}\tau$. We have:
\begin{equation*}
    \begin{aligned}
     & 1 - \Phi(t\tan\tfrac{\alpha}{2}) \le \Prob{}{h(\bx) = 1 \mid h(\by) = 1} \le 2(1 - \Phi(t\tan\tfrac{\alpha}{2})),                                              & 0 < \alpha\le\tfrac{\pi}{2}   \\
     & \frac{1}{4}(1 - \Phi(t\tan\tfrac{\alpha}{2}))^{\left(1 + \frac{2}{t^2}\right)^2}\le \Prob{}{h(\bx) = 1 \mid h(\by) = 1} \le 1 - \Phi(t\tan\tfrac{\alpha}{2}), & \tfrac{\pi}{2}\le\alpha < \pi
    \end{aligned}
\end{equation*}
\end{lemma}
Although technically interesting, the proof of \Cref{le:cond_prob} is not necessary for understanding the data structure. Those bounds serve primarily to provide a more precise analysis for each range of the parameters.
Interested readers can find the analysis in \Cref{apx:proof_cond_prob}.

\subparagraph*{Discussion.} As we mentioned earlier, variants of \Cref{le:cond_prob} are crucial ingredients of \lsh/\lsf\ approaches to ANN on the sphere. We briefly review previous contributions on this aspect, highlighting connections and differences with the analysis presented here.

Borrowing from previous work \cite{andoni2014beyond}, Andoni and Razenshteyn \cite{andoni2015optimal} used Gaussian random projectors as we do (up to scaling), but they focused on a specific value of the threshold, with a different goal in mind. In their Spherical \lsh\ scheme, random projectors following $\norm(\bzero, I_d)$ are used with a fixed threshold to partition $S^{d-1}$ using a (possibly not too large) set of sets carved out of  spherical caps with the same height. Considering the $1/\sqrt{d}$ scaling that we perform, their Theorem 3.1 corresponds to our \Cref{le:cond_prob} for $\tau = d^{-1/4}$ (see the claim of their Theorem 3.1 and its proof in their Appendix A). In particular, recall that two vectors on the unit sphere sharing an angle $\alpha$ have Euclidean distance $r = 2\sin\frac{\alpha}{2}$. If (i) we use \Cref{le:normal_tail_1} to approximate $1 - \Phi(t\tan\frac{\alpha}{2})$, (ii) replace $\sin\frac{\alpha}{2}$ with $r/2$ and, (iii) take the (natural) logarithm of the reciprocals of the upper and lower bounds provided by \Cref{le:cond_prob} for the joint probability, we obtain an expression almost identical to (2) and (3) from \cite[Thm. 3.1]{andoni2015optimal}, namely
\begin{equation*}
    \ln \left( \frac{1}{\Prob{}{\hash(\bx)  = 1\mid\hash(\by) = 1}} \right) = \frac{r^2}{4 - r^2}\frac{\sqrt{d}}{2} + O(\ln d).
\end{equation*}

Moreover, at least when $0\le\alpha\le\frac{\pi}{2}$, the ratio between upper and lower bound is constant, while the one in \cite[Theorem 3.1]{andoni2015optimal} is not, due to the presence of factors $1 \pm d^{-\Omega(1)}$ before the logarithm.
This said, the proofs and results given in \cite[Theorem 3.1]{andoni2015optimal} could be extended to general thresholds, in the end yielding asymptotic results on the parameter $\rho$ equivalent to those we will show in \Cref{subse:slsf_optimal}.
It is worth noting that the very same collision probability resurfaces in the analysis of the cross-polytope \lsh\ scheme proposed in \cite{andoni2015practical}, where the proof of the lower bound given in their Theorem 2 relies on computing a quantity $\Lambda(\tau, \eta)$.
This is a conditional probability, which turns out to be exactly $\Prob{}{\hash(\bx) = 1\mid\hash(\by) = 1}$ when $t = \eta$ and the angle $\alpha$ between $\bx$ and $\by$ satisfies $\cos\alpha = 1 - \frac{\tau^2}{2}$. Hence, \Cref{le:cond_prob} also provides a clean and tight characterization of this quantity.

The problem of computing the conditional probability from \Cref{le:cond_prob} has also been considered in a more general \lsf\ framework in \cite{christiani2017framework}. Here (see their Lemmas B.3 and B.4), they are interested in the joint probability that two standard normal variables $X$ and $Y$ with bounded correlation both exceed two different thresholds. The setting of our \Cref{le:cond_prob} corresponds to choosing $\lambda = 0$ in their framework. On the other hand, their proofs either do not cover all possible ranges for the angle (e.g., for the lower bound, when the angle exceeds $\pi/2$) or are only tight up to polynomial factors in (our) $t$ when specialized to our setting, since they rely on a previous result (Lemma B.2 in their appendix) estimating the probability of a multivariate normal vector belonging to a closed convex body, which provides a coarser approximation than \Cref{le:cond_prob}.

Finally, though \cite{becker2016new} introduced the idea of \lsf, their results are not straightforwardly equivalent to \Cref{le:cond_prob}, since the problem of estimating the probability of collision under random Gaussian projectors cannot be immediately rephrased as the geometrical problem of computing the volume of the intersection of spherical caps, though the two are clearly related due to concentration of the measure of $\norm(\bzero, I/d)$ on the surface of the unit sphere.

\subsection{Optimality of \eas}\label{subse:slsf_optimal}
In this section, we prove that the discriminating power of the spherical filters proposed in \Cref{sse:LSF} is asymptotically optimal, as measured by the parameter $\rho$ from \Cref{def:rho}.

\subparagraph*{Relation to \lsh.}
Before delving into the technical details, we provide some context that might hopefully benefit the reader, briefly discussing the reasons behind \Cref{def:rho} and the relationships with the analogous parameter $\rho$ defined for \lsh.
\lsf\ was introduced in \cite{becker2016new}, where they also defined $\rho$ as we do in this paper. The reason for their definition of $\rho$ stems from their analysis \cite[Section 3, Thm. 3.1]{becker2016new} of the time and memory costs of the data structure they propose (essentially, the same we introduced in \Cref{sse:LSF} up to the choice of the filters).
The authors of \cite{becker2016new} discuss the asymptotic behavior of the $\rho$ computed for the filters they propose, comparing it with the $\rho$ computed for the \lsh\ schemes proposed in \cite{andoni2015practical,andoni2015optimal}, though the two are clearly defined differently.
There is a practical justification for drawing this comparison though, since \cite{becker2016new} also makes a connection (similar to the one for \lsh) between the $\rho$ from \Cref{def:rho} and the performance of their \lsf-based data structure. Indeed, the connection between the $\rho$ from \lsf\ as defined in \Cref{def:rho} (and \cite{becker2016new}) and the analogous parameter for \lsh\ runs deeper.
In particular, Christiani proved \cite[Section 5, Lemma 5.1 and Thm. 1.4]{christiani2017framework} that every \lsf\ scheme for the $\ell_s$ distance with given value $\rho$ according to \Cref{def:rho} yields an \lsh\ scheme with the same value of the corresponding parameter $\rho$, up to lower order terms. This in particular implies that every \lsf\ scheme for $\ell_s$  must satisfy the general lower bound given in \cite[Thm. 3.1 and Cor. 3.2]{o2014optimal}, which is $1/c^2 - o_d(1)$ for the Euclidean distance.


We next prove (asymptotic) optimality of \eas.
\begin{theorem}\label{le:min_cap_rho}
Assume $\angle(\bx, \bq) = \gamma$, $\angle(\by, \bq) = c\gamma$ for $c > 1$ and $\tau$ is chosen so that $t = \sqrt{d}\tau$ satisfies  $t\tan\frac{\gamma}{2}\ge 1$. If $q_1 = \Prob{}{\hash(\bx) = 1 \mid \hash(\bq) = 1}$ and $q_2 = \Prob{}{\hash(\by) = 1 \mid \hash(\bq) = 1}$ then:
\[
    \rho = \frac{\log\frac{1}{q_1}}{\log\frac{1}{q_2}}\le\frac{1}{c^2} + o_{d}(1).
\]
\end{theorem}
The proof of \Cref{le:min_cap_rho} follows directly from \Cref{le:cond_prob} and is deferred to \Cref{apx:proof_min_cap_rho}.\\

We briefly note that $\tau$ is assumed to be a sufficiently large, fixed parameter to ensure that $\rho\rightarrow\frac{1}{c^2}$ as $d\rightarrow\infty$. In practice, $\tau = \omega(\frac{1}{\sqrt{d}})$ is enough for the same conclusions to apply.

\begin{remark}\label{rem:angle_vs_dist}
Since we decided to work directly with the angular distance, rather than the Euclidean distance on the unit sphere, one might wonder whether \Cref{le:min_cap_rho} implies optimality of Spherical \lsf\ with respect to the Euclidean distance on the sphere, which is the metric normally considered in the literature (e.g., \cite{andoni2015practical,andoni2015optimal}).
Consider two point pairs at angular distances $\gamma$ and $c\gamma$ respectively. The ratio $\hat{c}$ between the corresponding Euclidean distances is:
\[
    \hat{c} = \frac{\sin\frac{c\gamma}{2}}{\sin\frac{\gamma}{2}}\le c,
\]
where the inequality holds for all $0 < \gamma < \pi$. Hence, $\rho\le\frac{1}{c^2} + o_d(1)$ from \Cref{le:min_cap_rho} implies $\rho\le\frac{1}{\hat{c}^2} + o_d(1)$ and thus asymptotic optimality as $d\rightarrow\infty$.
\end{remark}

\subsection{Candidate and Space Complexity of the Basic Construction}\label{se:eas_anns}
We next analyze the performance of Spherical-\lsf\ in terms of the amount of extra memory needed to store the indexing structure and the cost of answering a query. Our analysis proceeds along the lines of the one presented in \cite[Thm. 3.4]{har2012approximate} for \lsh, albeit with the notable difference that in our case, we use the threshold $\tau$ as our lever to achieve desired (expected) space/time trade-offs. The following result holds:

\begin{theorem}\label{thm:main_thm}
Let $n$ be sufficiently large and $\delta$ be any constant in $ (0,1)$.  Assume  
$\tau \geq 
\sqrt{\frac{2\log n}{d \tan^2 \frac{c\gamma}{2}}}$ and $m = \left\lceil \frac{\ln (1/\delta)}{p_1} \right\rceil$, where $p_1$ is the joint probability that two points at angular distance $\gamma$ pass the same filter. Then, for any query $\bq$, Spherical-\lsf\ solves Angular $(\gamma, c\gamma)$-ANN by correctly answering with probability at least $1-\delta$ and performing $O(n^{\rho})$ angular distance computations in expectation.  
The expected size of the data structure (other than the space required to store the data points) is $O(n^{1+\rho})$.
\end{theorem}

\begin{proof}
    We define the following notation for the remainder of this proof: (i) $p_0 = \Prob{}{\hash(\bx) = 1}$, (ii) $p_1 = \Prob{}{\hash(\bx) = 1\wedge\hash(\bq) = 1\mid\angle(\bx, \bq) = \gamma}$, (iii) $p_2 = \Prob{}{\hash(\by) = 1\wedge\hash(\bq) = 1\mid \angle(\by, \bq) = c\gamma}$.
    To begin, note that if no point within distance $\gamma$ from $\bq$ exists, there is nothing to prove. Next, consider a data point $\bx \in \dset$ with angle $\angle(\bq,\bx) \leq \gamma$. 
    Observe that the query algorithm retrieves $\bx$ from a given bucket $B_i$ if $h_i(\bx) = h_i(\bq) = 1$. By definition, this happens with probability at least $p_1$. Since the $m$ random projectors are mutually independent,  the probability that $\bx$ is not found in any of the $m$ buckets is upper bounded by 
    $$
     \prod_{i=1}^{m}(1-p_1) = (1-p_1)^m \leq e^{-mp_1}\le \delta \, .
    $$
    

    In order to bound the query time of the data structure, we focus on the worst case in which a unique point $\bx \in \mathcal{X}$ exists, such that $\angle(\bq,\bx) \leq \gamma$, while $\angle(\bq,\by) \geq c\gamma$ for every $\by \in \mathcal{X} \setminus \{\bx\}$.
    
    Let $q_1 = p_1/p_0$ and $q_2 = p_2/p_0$ be the conditional probabilities that $\mathbf{x}$ and $\mathbf{y}$ are inserted in the filter given that $\mathbf{q}$ has been inserted, respectively. For constant failure probability $\delta$, $m = \Omega(\frac{1}{p_1}) = \Omega(\frac{1}{p_0q_1})$. 

    On average, \Cref{alg:lsf_query} examines $O(m p_2 n)$ points per query. Substituting $m$, the expected number of candidates is $O\left( \frac{q_2}{q_1} n \right)$.

    To achieve the target complexity, we seek a parameter $\tau$ such that $\frac{q_2}{q_1} n \leq n^\rho$, where $\rho = \frac{\log q_1}{\log q_2}$. Through algebraic manipulation, this inequality is satisfied when $q_2 \leq 1/n$.

    By \Cref{le:cond_prob}, $q_2$ is bounded by $2(1 - \Phi(t \tan \frac{c\gamma}{2}))$.
    Thus, the condition $q_2 \leq 1/n$ holds if
    
    $$
    t \geq \frac{\Phi^{-1}(1 - \frac{1}{2n})}{\tan \frac{c\gamma}{2}} \implies \tau \geq \frac{\Phi^{-1}(1 - \frac{1}{2n})}{\sqrt{d} \tan \frac{c\gamma}{2}}.
    $$
    
    Using the tail bound $1 - \Phi(z) \leq \frac{\phi(z)}{z}$ from \Cref{le:normal_tail_1}, we set $\frac{\phi(z)}{z} = \frac{1}{2n}$.
    Taking the logarithm, we obtain $\frac{z^2}{2} + \log z + \log \sqrt{2\pi} = \log(2n)$.
    For $n \to \infty$, the $z^2$ term dominates the logarithmic factors, yielding $z \approx \sqrt{2 \log n}$.
    Consequently, for large $n$, we can approximate the required threshold as
    
    $$
    \tau \geq \frac{\sqrt{2 \log n}}{\sqrt{d} \tan \frac{c\gamma}{2}}.
    $$

    Concerning the size of the data structure, the space complexity is dominated by the storage of the $m$ buckets.
    Since each of the $n$ points is stored in a bucket only if it passes the corresponding spherical filter, the expected total number of entries is $n \cdot m \cdot p_0$.
    Given that $m = O(1/(p_0 q_1))$, the average space complexity is
    
    $$
    O\left( n \cdot \frac{1}{p_0 q_1} \cdot p_0 \right) = O\left( \frac{n}{q_1} \right).
    $$

    By the definition of $\rho$, we have $q_1 = q_2^\rho$.
    Recalling our condition $q_2 \approx 1/n$, it follows that $1/q_1 \approx n^\rho$.
    Therefore, the average space complexity is $O(n^{1+\rho})$.
\end{proof}

\Cref{thm:main_thm} bounds the expected number of candidate points examined by the query algorithm, but it does not yet provide the same bound on the total query time.
Indeed, the direct implementation first evaluates all $m$ independently sampled filters in order to construct $\sig(\bq)$.
Although only $O(n^\rho)$ candidates are examined in expectation, the number $m$ of filter evaluations can be substantially larger and may become the dominant term.
In principle, this is important, since the query time should also include the number of filter evaluations that need to be performed. A similar problem arises (and is discussed) for the number of hash function computations in the \lsh\ literature (e.g., see \cite[Thm. 3.4]{har2012approximate} or \cite[Thm. 2.3]{andoni2015optimal}, where the authors assume $p_1, p_2\ge 1/n^{o_c(1)}$) and it can be addressed in a number of ways in \lsf. These include random product codes that modify the distribution from which filters are sampled \cite{becker2016new}, or powering and tensoring techniques  that amplify the locality sensitive properties of the filters \cite{christiani2017framework}. While these theoretical constructions allow to reduce the number of filter evaluations (e.g., bringing it down to $O(n^{\rho})$ in expectation), computing the signature of a query vector is rarely a bottleneck in practice (e.g., \cite{pham2022falconn++}). One of the reasons is that multiple inner products between independent standard normal vectors and a query vector can be effectively and very efficiently simulated using pseudo-random rotations, e.g., based on the Fast Hadamard Transform \cite{andoni2015practical,pham2022falconn++}.\\

In conclusion, the basic construction should be viewed as isolating the probabilistic components of Spherical \lsf.
In \Cref{se:tensoring} we address the remaining algorithmic bottleneck: efficiently identifying the filters (and hence the buckets) passed by a query.
We do so through the so-called \emph{tensoring techinque}, obtaining a large family of composite filters from a much smaller family of base filters.

\section{Efficient Filter Enumeration via Tensoring}\label{se:tensoring}
As we discussed at the end of \Cref{se:eas_anns}, 
the bottleneck of the query algorithm is the evaluation of the $m$ filters, which in general can dominate the average number of candidates $O(n^\rho)$.

\subsection{Tensoring}
In \cite{christiani2017framework}, the authors propose a technique called \emph{tensoring} to reduce the number of filters that need to be evaluated.
The idea is to \emph{simulate} $m$ filters from a smaller set of $m'$ filters, where $m' \ll m$.
In particular, we can construct $m = \binom{m'}{\eta}$ filters by taking all possible intersections of $\eta$ filters from the set of $m'$ filters.
More precisely, to solve the problem of the evaluation of the $m$ filters, the authors combine the tensoring technique with another technique called \emph{powering} (also known as \emph{AND-construction} in the \lsh\ literature), which consists in taking the intersection of $\kappa$ independent filters to obtain a new filter with improved locality-sensitive properties.

The combination of powering and tensoring allows to reduce the number of filters to evaluate for generic $(\gamma, c\gamma, q_1, q_2)$-\lsf\ families, for any values of $q_1$ and $q_2$.
In our case, however, the threshold $\tau$ already acts as an amplifier, similar to powering, so we do not need to use the powering technique.
This greatly simplifies both the data structure and the analysis, while still allowing to reduce the number of filters to evaluate.

We start defining the tensoring technique and its properties, and then we show how to apply it to our data structure.

\begin{definition}[Tensoring]\label{def:tensoring}
Let $\mathcal{F}$ be a $(\gamma, c\gamma, q_1, q_2)$-\lsf\ family, and let $m \in \mathbb{N}$ and $\eta \in [m]$.
Let $F$ be a collection of $m$ independent filters $h_1, \dots, h_m$ sampled uniformly at random from $\mathcal{F}$.
The tensoring of $F$ is the collection $F^{\otimes \eta}$ of $M = \binom{m}{\eta}$ filters, such that, for every $I \in \binom{[m]}{\eta}$ we define
\[
H_I(\bx) = \prod_{i \in I} h_i(\bx), \; \forall \bx \in S^{d-1}.
\]
Let $I_\bx = \{I \in \binom{[m]}{\eta} \mid H_I(\bx) = 1\}$, that is the set of tensor filters that $\bx$ passes.
The tensoring $F^{\otimes \eta}$ of $F$ has the following properties:
\begin{enumerate}[i)]
    \item if we consider two near points $\bx,\by \in S^{d-1}$ such that $\angle(\bx,\by) \leq \gamma$, and we set $m = \left\lceil \frac{\eta}{p_1} \right\rceil$, then $\Prob{}{\vert I_\bx \cap I_\by \vert \geq 1} \geq \frac{1}{2}$.
    \item if we consider two far points $\bx,\by \in S^{d-1}$ such that $\angle(\bx,\by) > c\gamma$, in expectation $\vert I_\bx \cap I_\by \vert \leq p_2^\eta M$.
    \item for every point $\bx \in S^{d-1}$, in expectation $\vert I_\bx \vert = p_0^\eta M$.
    \item the cost of computing $I_\bx$ in expectation is $O(m)$ plus the number of numerated $\eta$-subsets $I \in I_\bx$.
\end{enumerate}

\end{definition}
\begin{proof}
The last three properties are trivial.
For the first property, we define the binary random variable $X_i$ that is $1$ if $h_i(\bx) = h_i(\by) = 1$, and $0$ otherwise, for every $i = 1, \dots, m$.
We note that $p_{\bx,\by} = \Prob{}{X_i = 1} \geq p_1$.
Let us define the random variable $X = \sum_{i=1}^{m} X_i$ that indicates the number of filters that both $\bx$ and $\by$ passes, and observe that $X \sim \text{Bin}(m, p_{\bx,\by})$, with mean $\mu = p_{\bx,\by}m \geq p_1m \geq \eta$.
Since $\eta$ is an integer, we have that $\eta \leq \left\lfloor \mu \right\rfloor$.
Since the median $M_X$ of $X$ belongs to
$\{\lfloor\mu\rfloor,\lceil\mu\rceil\}$,\footnote{Recall that a median of $\text{Bin}(m,p)$ belongs to
$\{\lfloor mp\rfloor,\lceil mp\rceil\}$.} then we have $\eta \leq M_X$, and hence
\[
\Prob{}{\vert I_\bx \cap I_\by \vert \geq 1}
= \Prob{}{X \geq \eta} \geq \Prob{}{X \geq M_X} \geq \frac{1}{2}.
\]
\end{proof}

\subsection{The Tensorised Data Structure}
We are now ready to show how to modify the spherical-\lsf\ data structure using the tensoring technique in order to guarantee complexity $O(n^\rho)$.
As in \Cref{alg:lsf}, we sample $m$ independent Gaussian filters with parameter $\tau$ and compute $\sig(\bx)$.
After that, for each point $\bx \in \dset$ we consider all the $\eta$-tuple $I$ of indices of $\sig(\bx)$, and then we add the point $\bx$ into the bucket $B_I$.
Finally, we return all the non-empty buckets $B_I$.
Note that while in \Cref{alg:lsf_indexing} each point $\bx$ was inserted into $\vert \sig(x) \vert$ different buckets, in \Cref{alg:indexing_tensoring} each point $\bx$ is now inserted into $\binom{\vert \sig(\bx) \vert}{\eta}$ buckets.
Concerning the query, as for the indexing, we simply probe all the buckets $B_I$, for every $I \subseteq \sig(\bq)$ with $\vert I \vert = \eta$.
All pseudo-codes of the data structure are reported in \Cref{alg:lsf_tensoring,alg:indexing_tensoring,alg:query_tensoring}.

\noindent
\begin{minipage}[t]{0.46\textwidth}
    \begin{algorithm}[H]
      \caption{Tensorised Spherical-\lsf-Filter}\label{alg:lsf_tensoring}
      \textbf{Input:} a vector $\bx\in S^{d-1}$, a threshold $\tau \ge 0$, $m$ vectors $\theta_1, \dots, \theta_m \sim \norm(\bzero, I/d)$\\
      \textbf{Output} a set of indices $\sig(\bx)\subseteq [m]$
      \begin{algorithmic}[1]
        \State $\sig(\bx) = \{ i \in [m] \mid \theta_i^\top\bx \geq \tau \}$\;
        \State \Return $\sig(\bx)$\;
      \end{algorithmic}
    \end{algorithm}
\end{minipage}
\hfill 
\begin{minipage}[t]{0.46\textwidth}
    \begin{algorithm}[H] 
    \caption{Spherical-\lsf-Index}\label{alg:indexing_tensoring}
        \textbf{Input:} the set $\dset \subseteq S^{d-1}$, an integer $\eta$.\\
        \textbf{Output} a set of buckets $\mathcal{B}$
        \begin{algorithmic}[1]
            \For{$\bx \in \dset$}
                \State $I_\bx \gets \{I \subseteq \sig(\bx) \mid \vert I \vert = \eta\}$
                \For{$I \in I_\bx$}
                    $\buck_I \gets \buck_I \cup \{\bx\}$
                \EndFor
            \EndFor
            \State \Return $\mathcal{B} = \{B_I \mid I \in \binom{[m]}{\eta} \land B_I \neq \emptyset\}$
        \end{algorithmic}
    \end{algorithm}
\end{minipage}

\begin{center}
\begin{minipage}{0.6\textwidth}
    \begin{algorithm}[H]
    \caption{Spherical-\lsf-Query}\label{alg:query_tensoring}
        \textbf{Input:} a query $\bq$, a set of buckets $\mathcal{B}$, an integer $\eta$\\
        \textbf{Output:} near neighbor for $\bq$ or \texttt{not found}
        \begin{algorithmic}[1]
            \State $I_\bq \gets \{I \subseteq \sig(\bq) \mid \vert I \vert = \eta\}$
            \For{$I \in I_\bq$}
                \For{$\bx\in\buck_I$}
                    \If{$\angle(\bx, \bq)\le c\gamma$}
                        \State \Return $\bx$
                    \EndIf
                \EndFor
            \EndFor
            \State \Return \texttt{not found}
        \end{algorithmic}
    \end{algorithm}
\end{minipage}
\end{center}
\vspace{1em}

\begin{theorem}\label{thm:main_thm_tensoring}
Let $n$ be sufficiently large, $\eta = \left\lceil \frac{1}{\sin^2 \frac{\gamma}{2}} \right\rceil$, and assume $\tau \geq \sqrt{\frac{2\log n}{\eta d \tan^2 \frac{c\gamma}{2}}} + o(1)$ and $m = \left\lceil \frac{\eta}{p_1} \right\rceil$, where $p_1$ is the joint probability that two points at angular distance $\gamma$ pass the same filter with threshold parameter $\tau$.
Then, for every query $\bq \in S^{d-1}$, Spherical-\lsf\ solves Angular $(\gamma, c\gamma)$-ANN by correctly answering with probability at least $\frac{1}{2}$ and performing $O(n^{\rho + o(1)})$ angular-distance computations in expectation.  
The expected size of the data structure (other than the space required to store the data points) is $O(n^{1+\rho + o(1)})$.
\end{theorem}

\begin{proof}
By \Cref{def:tensoring}, we already know that with probability at least $\frac{1}{2}$ if there exists a point $\bx \in \dset$ such that $\angle(\bq,\bx) \leq \gamma$, then the Spherical-\lsf-Query of \Cref{alg:query_tensoring} will return it.

As for \Cref{thm:main_thm}, in order to bound the query time of the data structure, we focus on the worst case in which a unique point $\bx \in \dset$ exists, such that $\angle(\bq,\bx) \leq \gamma$, while $\angle(\bq,\by) \geq c\gamma$ for every $\by \in \dset \setminus \{ \bx \}$.\\

Query time is composed of three terms: (i) the number of candidates encountered when scanning the buckets, (ii) the number of buckets to probe, and (iii) the time required to calculate the
set of buckets $I_\bq$ to probe.

In the following we show that each of these terms is $O(n^{\rho + o(1)})$ in expectation, and that, therefore, the query time requires $O(n^{\rho + o(1)})$ angular distance computations in expectation.

\subparagraph*{Candidate estimation.}
Since $m = \left\lceil \frac{\eta}{p_1} \right\rceil$, by \Cref{def:tensoring} we know that in expectation the Spherical-\lsf-Query examines at most
\[
nMp_2^\eta = n \binom{m}{\eta}p_2^\eta \leq n \left( \frac{em}{\eta}\right)^\eta p_2^\eta = O(n) \left(\frac{p_2}{p_1}\right)^{\eta} = O(n) \left(\frac{q_2}{q_1}\right)^{\eta},
\]
where the third inequality holds since $e^\eta = O(1)$ (since $\eta$ is a constant), and the last inequality holds since, by definition, $p_1 = p_0q_1$ and $p_2 = p_0q_2$.
By definition of $\rho$, we have $q_1 = q_2^\rho$, hence $\frac{q_2}{q_1} = q_2^{1-\rho}$.
Now, we aim at a choice of $\tau$ such that $q_2^{(1-\rho)\eta} \leq n^{\rho-1}$.
By taking the logarithms it holds when
\begin{align*}
    (1-\rho)\eta \log{q_2} &\leq (\rho - 1) \log{n}
    \;\;\; \implies \;\;\; (1-\rho)\eta \log{\frac{1}{q_2}} \geq (1- \rho) \log{n}\\
    \implies \log{\frac{1}{q_2}} &\geq \frac{\log{n}}{\eta}.
\end{align*}
By \Cref{le:cond_prob} and \Cref{le:normal_tail_1}, we know that, for sufficiently large $d$, $\log{\frac{1}{q_2}} = \frac{t^2}{2} \tan^2\frac{c\gamma}{2} + O(\log{t})$, therefore for constant $\eta$
\begin{equation}\label{eq:tau_choice_tensoring}
t \geq \sqrt{\frac{2\log{n}}{\eta \tan^2\frac{c\gamma}{2}} + O(\log{t})} \implies \tau \geq \sqrt{\frac{2\log{n}}{\eta d \tan^2\frac{c\gamma}{2}}} + o_d(1).
\end{equation}
Therefore, we have that the expected number of candidates that Spherical-\lsf-Query examines is at most 
\[
O(n)q_2^{(1-\rho)\eta} \leq O(n) n^{\rho - 1} = O(n^{\rho}).
\]

\subparagraph{Bucket probes.}
By \Cref{def:tensoring}, the expected number of buckets $\buck_I \in \mathcal{B}$ that Spherical-\lsf-Query will probe is
\[
\Expec{}{\vert I_\bq \vert} = Mp_0^\eta \leq \left( \frac{em}{\eta} \right)^\eta p_0^\eta = O(1) \left(\frac{p_0}{p_1} \right)^\eta = O(1) \left(\frac{1}{q_1} \right)^\eta = O(1) q_2^{-\eta\rho}.
\]
Since we know that $q_2^\eta \leq n^{-1}$, we obtain the desired bound $\Expec{}{\vert I_\bq \vert} = O(n^\rho)$.

\subparagraph{Generation of $I_\bq$.}
By \Cref{def:tensoring} we know that $I_\bq$ can be computed in time $O(dm + \vert I_\bq \vert)$.
Since we already know that $\vert I_\bq \vert$ in expectation is $O(n^{\rho})$, it remains to show that $m = O(n^{\rho + o(1)})$.
By definition $m = \left\lceil \frac{\eta}{p_1} \right\rceil$.
Observe that, since $\eta$ is a constant, the dominant term is $\frac{1}{p_1}$.
By taking the logarithm we have
\begin{equation}\label{eq:number_buckets_tensoring}
    \log\frac{1}{p_1} = \log\frac{1}{q_1p_0} = \log\frac{1}{p_0} + \log\frac{1}{q_1} = \frac{t^2}{2} + \frac{t^2}{2}\tan^2\frac{\gamma}{2} + O(\log{t}) = \frac{1+\tan^2\frac{\gamma}{2}}{\eta \tan^2\frac{c\gamma}{2}}\log{n} + o(\log{n}),
\end{equation}
where the last inequality follows from the choice of $t$ from \Cref{eq:tau_choice_tensoring}.
By definition of $\eta$ we have
\[
\eta \geq \frac{1}{\sin^2\frac{\gamma}{2}} = \frac{\cos^2\frac{\gamma}{2} + \sin^2\frac{\gamma}{2}}{\sin^2\frac{\gamma}{2}} = \frac{1}{\tan^2\frac{\gamma}{2}} + 1 = \frac{\tan^2\frac{\gamma}{2} + 1}{\tan^2\frac{\gamma}{2}}.
\]
By substituting $\eta$ in \Cref{eq:number_buckets_tensoring} we obtain the bound
\[
\log\frac{1}{p_1} \leq \frac{\tan^2\frac{\gamma}{2}}{\tan^2\frac{c\gamma}{2}}\log{n} (1+o(1)).
\]
Finally, we observe that 
\[
\rho = \frac{\log\frac{1}{q_1}}{\log\frac{1}{q_2}} = \frac{\tan^2\frac{\gamma}{2}}{\tan^2\frac{c\gamma}{2}} + o(1).
\]
This implies that $\log\frac{1}{p_1} \leq (\rho + o(1))\log{n}$, that is $m = O(n^{\rho + o(1)})$.\\

Concerning the size of the data structure, the space complexity is dominated by the storage of the buckets $\mathcal{B}$.
Since each point $\bx \in \dset$ is stored in $\vert I_\bx \vert$ buckets,  the expected total number of entries in the data structure is $ \sum_{\bx \in \dset} \Expec{}{\vert I_\bx \vert} = n M p_0^\eta$.
From before we already know that $Mp_0^\eta = O(n^\rho)$.
Therefore, the data structure stores $O(n^{1+\rho})$ points in expectation.
\end{proof}
\section{Beyond Symmetric Filters: Space-Time Tradeoffs}
\label{sec:asymmetric_lsf}

The symmetric filters studied in the previous sections provide a balanced solution, in which the same buckets are associated with queries and data points.
The general \lsf\ framework removes this restriction by representing a random filter as a pair $(Q,U)$, sampled from a distribution over pairs of subsets of the underlying space.
Here, $Q$ is the \emph{query region}, determining whether a bucket is inspected by a query, while $U$ is the \emph{update region}, determining whether a data point is stored in that bucket.

More precisely, a filter family is described by parameters $p_1,p_2,p_q,p_u$ such that
\begin{itemize}
    \item $\Prob{(Q,U)}{\bq \in Q,\ \bx \in U} \geq p_1$ if $\bq$ and $\bx$ are nearby points;
    \item $\Prob{(Q,U)}{\bq \in Q,\ \bx \in U} \leq p_2$ if $\bq$ and $\bx$ are distant points;
    \item $\Prob{(Q,U)}{\bx \in Q} \leq p_q$ and $\Prob{(Q,U)}{\bx \in U} \leq p_u$ for every point $\bx$.
\end{itemize}
The use of different query and update regions gives rise to the two exponents
\[
\rho_q =\frac{\log(p_q/p_1)}{\log(p_q/p_2)},\qquad
\rho_u = \frac{\log(p_u/p_1)}{\log(p_q/p_2)}.
\]
The first exponent governs the query complexity, while the second governs the update complexity and the additional space usage.
Consequently, asymmetric filters make it possible to redistribute the computational cost between query time and space.
In total, the tradeoff is $n^{\rho_q + o(1)}$ distance computations for the query time, and $n^{1 + \rho_u + o(1)}$ additional space.

In the spherical setting, this asymmetry admits a direct geometric interpretation.
Given a common Gaussian projector $\theta_i$, one may define two spherical caps with different thresholds:
\[
Q_i = \left\{ \bq \in S^{d-1} \mid \theta_i^\top \bq \geq \tau_q \right\},\qquad
U_i = \left\{ \bx \in S^{d-1} \mid \theta_i^\top \bx \geq \tau_u \right\}.
\]
The threshold $\tau_q$ controls how frequently a query probes a bucket, whereas $\tau_u$ controls how frequently a data point is stored in a bucket.
Varying these thresholds independently therefore yields different query-space tradeoffs\footnote{Note that, if $\tau_q \leq \tau_u$ then we have that $U_i \subseteq Q_i$, and vice versa.}.

The symmetric Gaussian filters considered throughout this paper correspond to the special choice $\tau_q = \tau_u = \tau$.
In this case, $Q_i = U_i$, the two marginal probabilities coincide, $p_q = p_u = p_0$, and the two exponents collapse to the single balanced parameter
\[
\rho_q = \rho_u = \frac{\log(p_0/p_1)}{\log(p_0/p_2)} = \frac{\log(1/q_1)}{\log(1/q_2)} = \rho.
\]
Our restriction to this balanced setting is deliberate.
It allows us to explain the geometry of spherical filters, the role
of the threshold $\tau$, the optimal exponent $\rho$, and the tensoring construction without carrying two sets of parameters throughout the analysis.
In conclusion, asymmetric filters preserve the same underlying mechanism, but simply use different query and update regions to redistribute computational cost between query time and space.

\subparagraph{Beyond data-independent filters.}
\label{subse:beyond_data_independent}
The space-time tradeoff described thus far remains data-independent: the distribution of the filters is fixed without inspecting the input dataset.
A separate and more powerful generalization allows the filtering or partitioning strategy to depend on the dataset itself.
In this data-dependent setting, stronger exponents become possible.

For Euclidean ANN, Andoni~et~al.~\cite{andoni2017optimal} established a tight space-time tradeoff for all those exponents $\rho_q \geq 0, \rho_u \geq 0$ satisfying
\[
c^2\sqrt{\rho_q} + (c^2-1)\sqrt{\rho_u} = \sqrt{2c^2-1},
\]
showing the existence of a data structure with query time $dn^{\rho_q + o(1)}$ and an additional space footprint of $n^{1+\rho_u + o(1)}$.

At the balanced point $\rho_q = \rho_u = \rho$, the tradeoff yields the exponent
\[
\rho = \frac{1}{2c^2-1}.
\]
This is strictly smaller than the balanced data-independent exponent $1/c^2$.
It is important to note that there is no contradiction with the lower bound discussed in \Cref{subse:slsf_optimal}: the two bounds concern different models.
The bound $1/c^2$ applies to the data-independent symmetric setting considered in the main body of this paper, whereas the stronger exponent above is achieved by allowing the data structure to adapt to the input dataset.


\bibliography{biblio}

\newpage
\appendix

\section{Deferred Proofs}



\subsection{Proof of \Cref{fa:standard_ratio}}\label{se:apx_prel}

\begin{proof}
    Assume $\eta_2\ge\eta_1\ge 0$. We have:
    \[
        \frac{1}{\sqrt{2\pi}}\int_{\eta_2}^{+\infty}e^{-\frac{x^2}{2}}dx = \frac{1}{\sqrt{2\pi}}\int_{\eta_1}^{+\infty}e^{-\frac{(z + \eta_2 - \eta_1)^2}{2}}dz,
    \]
    where we performed the change of variable $z = x + \eta_1 - \eta_2$.

    We further have $(z + \eta_2 - \eta_1)^2 = z^2 + \eta_2^2 + 2(\eta_2 - \eta_1)z 
    - 2\eta_1\eta_2 + \eta_1^2$. Moreover, $2(\eta_2 - \eta_1)z - 2\eta_1\eta_2 + \eta_1^2\ge -\eta_1^2$ whenever $z\ge\eta_1$, since $\eta_2 - \eta_1\ge 0$. As a consequence, $(z + \eta_2 - \eta_1)^2\ge z^2 + \eta_2^2 - \eta_1^2$ for $z\in [\eta_1, +\infty]$, whence:
    \[
        \frac{1}{\sqrt{2\pi}}\int_{\eta_1}^{+\infty}e^{-\frac{(z + \eta_2 - \eta_1)^2}{2}}dz\le\frac{e^{-\frac{\eta_2^2 - \eta_1^2}{2}}}{\sqrt{2\pi}}\int_{\eta_1}^{+\infty}e^{-\frac{z^2}{2}}dz,
    \]
    which proves the claim.
\end{proof}

\subsection{Proof of \Cref{le:cond_prob}}\label{apx:proof_cond_prob}

\begin{proof}
    We set $X = \sqrt{d}\theta^\top\bx$ and $Y = \sqrt{d}\theta^\top\by$ in the remainder of this proof.\footnote{It may be interesting to note that $(X, Y)$ is a bivariate normal variable with covariance $\cos\alpha$.} Due to rotational invariance of the (multivariate) standard normal distribution, we can assume without loss of generality that $\be_1 = \by$, while $\be_2$ is orthogonal to $\be_1$ (and thus $\by$), belongs to the plane spanned by $\bx$ and $\by$ and forms an acute angle with \(\bx\). Finally, setting $Z = \sqrt{d}\theta_2$, we have $Y = \sqrt{d}\theta_1$ and $X = Y\cos\alpha + Z\sin\alpha$, with $Y, Z\sim\norm(0, 1)$ and independent. Next, we have:
    \begin{equation}\label{eq:tail_3}
    \begin{aligned}
        & \Prob{}{h(\bx) = 1\wedge h(\by) = 1} = \Prob{}{\theta^\top\bx\ge\tau\wedge\theta^\top\by\ge\tau} = \Prob{}{X\ge t\wedge Y\ge t}\\
        & = \Prob{}{Y\cos\alpha + Z\sin\alpha\ge t\wedge Y\ge t}.
    \end{aligned}
    \end{equation}
    We next continue with $\Prob{}{Y\cos\alpha + Z\sin\alpha\ge t\wedge Y\ge t}$, since $\Prob{}{h(\by) = 1} = \Prob{}{Y\ge t} = 1 - \Phi(t)$. We have:
    \begin{equation}\label{eq:tail_1}
    \begin{aligned}
        &\Prob{}{Y\cos\alpha + Z\sin\alpha\ge t\wedge Y\ge t} = \int_{t}^{\infty}\Prob{}{Y\cos\alpha + Z\sin\alpha \geq t \mid Y = v} \phi(v) \, dv\\
        & = \int_{t}^{\infty}\Prob{}{Z \geq \frac{t - Y\cos\alpha}{\sin\alpha}\mid Y = v} \phi(v) \, dv = \int_{t}^{\infty}\Prob{}{Z \geq \frac{t - v\cos\alpha}{\sin\alpha}} \phi(v) \, dv,
    \end{aligned}
    \end{equation}
    where we recall that
    $$
    \Prob{}{Z \geq \frac{t - v\cos\alpha}{\sin\alpha}} = 1 - \Phi\left(\frac{t - v\cos\alpha}{\sin\alpha}\right).
    \footnote{We note that \eqref{eq:tail_1} can be equivalently obtained by observing that the bivariate \((X, Y)\) follows a normal distribution with covariance matrix $\Sigma_{XY} = \begin{pmatrix}
                        1          & \cos\alpha \\
                        \cos\alpha & 1
                    \end{pmatrix}$.}
    $$
    Next, the proof proceeds by considering the cases $\alpha \in [0,\frac{\pi}{2}]$ and $\alpha \in [\frac{\pi}{2}, \pi)$ separately, since the events $(X\ge t)$ and $(Y\ge t)$ are positively (resp. negatively) correlated in the two cases. This makes the upper bound harder to prove in the first case and the lower bound harder to prove in the second.

    \subparagraph*{Case 1: $\alpha \in [0,\frac{\pi}{2}]$.}
    
    We begin by proving the lower bound, which is simpler in this case and can be derived directly from \eqref{eq:tail_3}. In particular:
    \begin{equation*}
    \begin{aligned}
        & \Prob{}{Y\cos\alpha + Z\sin\alpha\ge t\wedge Y\ge t} = \Prob{}{Z\sin\alpha\ge t - Y\cos\alpha\wedge Y\ge t}\\
        & \ge\Prob{}{Z\sin\alpha\ge t(1 - \cos\alpha)\wedge Y\ge t},
    \end{aligned}
    \end{equation*}
    since the event $(Z\sin\alpha\ge t(1 - \cos\alpha)\wedge Y\ge t)$ implies the event $(Z\sin\alpha\ge t - Y\cos\alpha\wedge Y\ge t)$.
    We therefore have:
    \begin{equation*}
    \begin{aligned}
        & \Prob{}{h_i(\bx) = 1 \mid h_i(\by) = 1} = \frac{\Prob{}{h_i(\bx) = 1 \wedge h_i(\by) = 1}}{\Prob{}{h_i(\by) = 1}}
        \ge\frac{\Prob{}{Z\sin\alpha\ge t(1 - \cos\alpha)\wedge Y\ge t}}{\Prob{}{Y\ge t}}\\
        & =\Prob{}{Z\sin\alpha\ge t(1 - \cos\alpha)} = \Prob{}{Z\ge\ t\frac{1 - \cos\alpha}{\sin\alpha}}
        =\Prob{}{Z\ge t\tan\frac{\alpha}{2}} = 1 - \Phi\left(t\tan\frac{\alpha}{2}\right),
    \end{aligned}
    \end{equation*}
    where the third equality follows since the events $(Z\sin\alpha\ge t(1 - \cos\alpha))$ and $(Y\ge t)$ are independent, the fifth equality follows since $(1 - \cos\alpha)/\sin\alpha = \tan\frac{\alpha}{2}$, while the sixth follows since $Z\sim\norm(0, 1)$.

    We next prove the upper bound, which requires more work.
    We introduce $a = \cos\frac{\alpha}{2}$ and $b = \sin\frac{\alpha}{2}$. Since $(X, Y)$ is a centered Gaussian vector with covariance matrix
    \[
        \begin{pmatrix}
            1 & \cos\alpha \\
            \cos\alpha & 1
        \end{pmatrix}
        =
        \begin{pmatrix}
            1 & a^2 - b^2 \\
            a^2 - b^2 & 1
        \end{pmatrix},
    \]
    it has the same distribution as $(aU + bV, aU - bV)$, where $U, V\sim\norm(0, 1)$ are independent. Therefore:
    \begin{equation*}
    \begin{aligned}
        & \Prob{}{\hash(\bx) = 1\wedge\hash(\by) = 1} = \Prob{}{aU + bV\ge t \wedge aU - bV\ge t}\\
        & = \Prob{}{aU\ge t + b|V|} = 2\int_0^\infty \Prob{}{U\ge \frac{t + bv}{a}}\phi(v)\,dv.
    \end{aligned}
    \end{equation*}
    Next, for every $\eta\ge t\ge 0$ we have $1 - \Phi(\eta)\le e^{-\frac{\eta^2 - t^2}{2}}(1 - \Phi(t))$ from \Cref{fa:standard_ratio}.
    Applying this inequality with $\eta = \frac{t + bv}{a}$, dividing by $\Prob{}{Y\ge t} = 1 - \Phi(t)$ and recalling that $U$ and $V$ are independent, we obtain:
    \begin{equation*}
    \begin{aligned}
        & \Prob{}{\hash(\bx) = 1 \mid \hash(\by) = 1} \le 2\int_0^\infty e^{-\frac{\left(\frac{t + bv}{a}\right)^2 - t^2}{2}}\phi(v)\,dv\\
        & = 2\int_0^\infty \frac{1}{\sqrt{2\pi}}e^{-\frac{\left(\frac{t + bv}{a}\right)^2 - t^2 + v^2}{2}}\,dv
        = 2\int_0^\infty \frac{1}{\sqrt{2\pi}}e^{-\frac{(v + tb)^2}{2a^2}}\,dv\\
        & = 2a\int_{tb/a}^{\infty}\phi(z)\,dz = 2a\left(1 - \Phi\left(t\tan\frac{\alpha}{2}\right)\right)
        \le 2\left(1 - \Phi\left(t\tan\frac{\alpha}{2}\right)\right),
    \end{aligned}
    \end{equation*}
    where the third equality follows from $a^2 + b^2 = 1$ and the fourth from the change of variable $z = \frac{v + tb}{a}$.

    \subparagraph*{Case 2: $\alpha \in [\frac{\pi}{2},\pi)$.}
    We begin with the lower bound, which is harder to obtain, since $X$ and $Y$ are negatively correlated when $\alpha > \frac{\pi}{2}$, since $\cos\alpha < 0$. We use again \eqref{eq:tail_1} to write:
    \[
        \Prob{}{h(\bx) = 1\wedge h(\by) = 1} = \int_{t}^{\infty}\Prob{}{Z \geq \frac{t - v\cos\alpha}{\sin\alpha}} \phi(v) \, dv.
    \]
    Finding a tight lower bound to the probability above requires some care, due to the negative correlation between $X$ and $Y$ when $\alpha > \frac{\pi}{2}$ and, therefore, $\cos\alpha < 0$. To begin, we prove that the conditional expectation $\Expec{}{Y - t\mid Y\ge t}$ is $O(\frac{1}{t})$, which allows us to consider $v$ deterministically close to $t$ in the integral above. To this purpose \Cref{fa:expec} implies:
    \begin{align*}
        & \Expec{}{Y - t\mid Y\ge t} = \frac{1}{1 - \Phi(t)}\int_t^{\infty}(v - t)\phi(v)dv = \frac{\phi(t) - t(1 - \Phi(t))}{1 - \Phi(t)}.
    \end{align*}
    But $\phi(t)\le\frac{t^2 + 1}{t}(1 - \Phi(t))$ from \Cref{le:normal_tail_1}, whence $\phi(t) - t(1 - \Phi(t))\le\frac{1 - \Phi(t)}{t}$, so that $\Expec{}{Y - t\mid Y\ge t}\le\frac{1}{t}$.
    An application of Markov's inequality allows us to conclude that:
    \[
        \Prob{}{Y > t + \frac{2}{t} \mid Y\ge t} = \Prob{}{Y - t> \frac{2}{t}\mid Y\ge t}\le\frac{1}{2},
    \]
    yielding:
    \begin{equation}\label{eq:markov}
        \Prob{}{t\le Y\le t + \frac{2}{t}}\ge\frac{1}{2}(1 - \Phi(t)).
    \end{equation}
    We continue with:
    \begin{equation*}
    \begin{aligned}
        & \Prob{}{h(\bx) = 1\wedge h(\by) = 1} \geq \int_t^{t + \frac{2}{t}} \Prob{}{Z \geq \frac{t - v\cos\alpha}{\sin\alpha}}\phi(v)dv\\
        &\geq \Prob{}{Z \geq \frac{t - \left(t + \frac{2}{t}\right)\cos\alpha}{\sin\alpha}}\int_t^{t + \frac{2}{t}}\phi(v)dv= \left( 1 - \Phi\left(\frac{t - \left(t + \frac{2}{t}\right)\cos\alpha}{\sin\alpha}\right) \right)\int_t^{t + \frac{2}{t}}\phi(v)dv \\
        & \ge\frac{1}{2} \left(1 - \Phi\left(\frac{t - \left(t + \frac{2}{t}\right)\cos\alpha}{\sin\alpha}\right)\right)(1 - \Phi(t))\ge\frac{1}{2} \left(1 - \Phi\left(\left(t + \frac{2}{t}\right)\tan\frac{\alpha}{2}\right)\right)(1 - \Phi(t)),
    \end{aligned}
    \end{equation*}
    where the fourth inequality follows from \eqref{eq:markov}, while the last follows since \(1 - \Phi(x)\) is decreasing.
    By \Cref{le:normal_tail_1}, we have that $\frac{\phi(x)}{2x} \leq \frac{x}{x^2+1}\phi(x) \leq 1-\Phi(x) \leq \frac{\phi(x)}{x}$.
    Hence
    \begin{equation*}
    \begin{aligned}
        \frac{1 - \Phi\left( \left(1 + \frac{2}{t^2}\right) t \tan \frac{\alpha}{2} \right)}{1 - \Phi\left( t \tan \frac{\alpha}{2} \right)^{\left(1 + \frac{2}{t^2}\right)^2}}
        &\geq \frac{\phi\left( \left(1 + \frac{2}{t^2}\right) t \tan \frac{\alpha}{2} \right)}{\phi\left(t \tan \frac{\alpha}{2} \right)^{\left(1 + \frac{2}{t^2}\right)^2}} \frac{\left(t \tan \frac{\alpha}{2} \right)^{\left(1 + \frac{2}{t^2}\right)^2-1}}{2 \left(1 + \frac{2}{t^2}\right) }\\
        &= (2\pi)^{\frac{\left(1 + \frac{2}{t^2}\right)^2-1}{2}} \frac{\left(t \tan \frac{\alpha}{2} \right)^{\left(1 + \frac{2}{t^2}\right)^2-1}}{2 \left(1 + \frac{2}{t^2}\right) } \geq \frac{1}{2}.
    \end{aligned}
    \end{equation*}
    Where last inequality holds since $\left(t \tan \frac{\alpha}{2} \right)^{\left(1 + \frac{2}{t^2}\right)^2-1} \geq 1$, and $(2\pi)^{\frac{\left(1 + \frac{2}{t^2}\right)^2-1}{2}} \geq \left(1 + \frac{2}{t^2}\right)$ (we remind the assumption $t \tan\frac{\alpha}{2} \geq 1$, and observe that $t \geq 1$ for $\alpha \in [\frac{\pi}{2}, \pi)$).
    This implies
    \[
    \Prob{}{h(\bx) = 1\mid h(\by) = 1} \geq \frac{1}{4}(1 - \Phi(t\tan\tfrac{\alpha}{2}))^{\left(1 + \frac{2}{t^2}\right)^2}.
    \]

    The upper bound is easier in this case, due to the negative correlation between $X$ and $Y$ when $\frac{\pi}{2}\le\alpha < \pi$. Using again \eqref{eq:tail_1} we have:
    \begin{equation*}
    \begin{aligned}
        & = \Prob{}{\hash(\bx) = 1 \mid \hash(\by) = 1}
        = \frac{\int_{t}^{\infty} \left(1 - \Phi\left(\frac{t - v\cos\alpha}{\sin\alpha}\right)\right)\phi(v)\,dv}{\Prob{}{Y \geq t}}\leq \frac{\left(1 - \Phi\left(\frac{t - t\cos\alpha}{\sin\alpha}\right)\right)\int_{t}^{\infty}\phi(v)\,dv}{\Prob{}{Y \geq t}} \\
        & = 1 - \Phi\left(t\tan\frac{\alpha}{2}\right),
    \end{aligned}
    \end{equation*}
    where the second inequality follows since $1 - \Phi(x)$ is decreasing and $\cos\alpha < 0$ implies $t - t\cos\alpha\le t - v\cos\alpha$, while the last equality follows since $\Prob{}{Y \geq t} = \int_{t}^{\infty}\phi(v)\,dv$ and recalling that $\frac{1 - \cos\alpha}{\sin\alpha} = \tan\frac{\alpha}{2}$.
    This concludes the proof.
\end{proof}

\subsection{Proof of \Cref{le:min_cap_rho}}\label{apx:proof_min_cap_rho}

\begin{proof}
    The proof proceeds by applying \Cref{le:cond_prob} to three cases.

    \subparagraph*{Case 1: $0 < \gamma < c\gamma \le\frac{\pi}{2}$.} In this case, \Cref{le:cond_prob} gives:
    \begin{equation}\label{eq:p_1_and_p2}
        q_1\ge 1-\Phi(t \tan\tfrac{\gamma}{2}) \text{ and }
        q_2\le 2(1-\Phi(t \tan\tfrac{c\gamma}{2})).
    \end{equation}
    Next:
    \[
        \ln\frac{1}{q_1}\le\ln\frac{1}{1 - \Phi(t\tan\frac{\gamma}{2})}
        \le\frac{t^2}{2}\tan^2\frac{\gamma}{2} + \ln\left(2t\tan\frac{\gamma}{2}\right) + \frac{1}{2}\ln(2\pi),
    \]
    where the second inequality follows from \eqref{eq:p_1_and_p2}, the lower bound to $1 - \Phi(t\tan\frac{\gamma}{2})$ granted by \Cref{le:normal_tail_1} and recalling that we assume $t\tan\frac{\gamma}{2}\ge 1$. Moreover:
    \begin{equation*}
    \begin{aligned}
        &\ln\frac{1}{q_2}\ge\ln\frac{1}{2(1-\Phi(t \tan\frac{c\gamma}{2}))}\ge c^2\frac{t^2}{2}\tan^2\frac{\gamma}{2} + \ln\left(t\tan\frac{\gamma}{2}\right) + \frac{1}{2}\ln(2\pi) - \ln 2\\
        &> c^2\frac{t^2}{2}\tan^2\frac{\gamma}{2} + \ln\left(t\tan\frac{\gamma}{2}\right),
    \end{aligned}
    \end{equation*}
    where the first inequality follows from \eqref{eq:p_1_and_p2}, while the second follows from the upper bound to $1 - \Phi(t\tan\frac{c\gamma}{2})$ granted by \Cref{le:normal_tail_1} and from convexity of $\tan\alpha$ for $0\le\alpha < \frac{\pi}{2}$. As a result:
    \[
        \rho = \frac{\ln\frac{1}{q_1}}{\ln\frac{1}{q_2}}\le\frac{1}{c^2} + \frac{\ln\left(2t\tan\frac{\gamma}{2}\right) + \frac{1}{2}\ln(2\pi)}{c^2\frac{t^2}{2}\tan^2\frac{\gamma}{2} + \ln(t\tan\frac{\gamma}{2})} = \frac{1}{c^2} + o_{d}(1),
    \]
    where we recall that $t = \sqrt{d}\tau$, with $\tau$ a fixed parameter (i.e. constant w.r.t. $d$).

    \subparagraph*{Case 2: $0 < \gamma\le\frac{\pi}{2} < c\gamma < \pi$.} This case is similar to the previous one. In particular, the lower bound to $q_1$ is the same, while the upper bound to $q_2$ becomes
    \[
        q_2\le 1-\Phi(t \tan\tfrac{c\gamma}{2})
    \]
    from \Cref{le:cond_prob} (case $\frac{\pi}{2}\le\alpha < \pi$), yielding:
    \[
        \ln\frac{1}{q_2}\ge c^2\frac{t^2}{2}\tan^2\tfrac{\gamma}{2} + \ln(t\tan\tfrac{\gamma}{2}) + \frac{1}{2}\ln(2\pi) > c^2\frac{t^2}{2}\tan^2\frac{\gamma}{2} + \ln(t\tan\tfrac{\gamma}{2}) + \ln2,
    \]
    so that we again have $\rho\le\frac{1}{c^2} + o_{d}(1)$.

    \subparagraph*{Case 3: $\frac{\pi}{2} < \gamma < c\gamma < \pi$.} In this case, we have
    \[
        q_1\ge\frac{1}{4}(1 - \Phi(t\tan\tfrac{\alpha}{2}))^{\left(1 + \frac{2}{t^2}\right)^2}
    \]
    from \Cref{le:cond_prob} (lower bound for $\frac{\pi}{2}\le\alpha < \pi$), while the upper bound on $p_2$ is the same as the previous case, so that we have:
    \[
    \ln\frac{1}{q_1}\le\ln\frac{4}{(1 - \Phi(t\tan\frac{\gamma}{2}))^{\left(1 + \frac{2}{t^2}\right)^2}}\le\left(1 + \frac{2}{t^2}\right)^2\left(\frac{t^2}{2}\tan^2\frac{\gamma}{2} + \ln(2t\tan\frac{\gamma}{2})\right) + \ln 4
    \]
    and
    \[
    \ln\frac{1}{q_2}\ge c^2\frac{t^2}{2}\tan^2\frac{\gamma}{2} + \ln(t\tan\frac{\gamma}{2}) + \ln 2.
    \]
    As a result:
    \begin{equation*}
    \begin{aligned}
        &\rho\le\frac{\left(1 + \frac{2}{t^2}\right)^2\left(\frac{t^2}{2}\tan^2\frac{\gamma}{2} + \ln(2t\tan\frac{\gamma}{2})\right) + \ln 4}{c^2\frac{t^2}{2}\tan^2\frac{\gamma}{2} + \ln(t\tan\frac{\gamma}{2}) + \ln 2} < \left(1 + \frac{2}{t^2}\right)^2\frac{1}{c^2} + \frac{\left(1 + \frac{2}{t^2}\right)^2\ln(2t\tan\frac{\gamma}{2}) + \ln 4}{c^2\frac{t^2}{2}\tan^2\frac{\gamma}{2} + \ln(t\tan\frac{\gamma}{2}) + \ln 2}\\
        &= \frac{1}{c^2} + \left(\frac{4}{t^2} + \frac{4}{t^4}\right)\frac{1}{c^2} + \frac{\left(1 + \frac{2}{t^2}\right)^2\ln(2t\tan\frac{\gamma}{2}) + \ln 4}{c^2\frac{t^2}{2}\tan^2\frac{\gamma}{2} + \ln(t\tan\frac{\gamma}{2}) + \ln 2} = \frac{1}{c^2} + o_{d}(1),
    \end{aligned}
    \end{equation*}
    where we again recall that $t = \sqrt{d}\tau$, with $\tau$ a fixed parameter. This concludes the proof.
\end{proof}

\end{document}